\documentclass[review]{elsarticle}

\usepackage{lineno}
\usepackage{graphicx}
\usepackage{bm}

\usepackage{algorithm}
\usepackage{multirow}
\usepackage{caption}
\usepackage{float}
\usepackage[noend]{algpseudocode}

\usepackage{amsthm}
\usepackage{amssymb}
\usepackage{amsmath}
\usepackage[toc,page]{appendix}
\usepackage{array}
\usepackage{enumitem}
\usepackage{tikz}
\usetikzlibrary{shapes.geometric,calc, angles,patterns}
\usepackage{pgflibraryarrows}
\usepackage{pgflibrarysnakes}

\usepackage{hyperref}
\usepackage{cleveref}
\usepackage{xcolor}
\modulolinenumbers[5]

\newcommand{\MER}{\textsc{MER}}
\newcommand{\MEP}{\textsc{MEP}}
\newcommand{\CH}{\textsc{CH}}

\DeclareMathOperator{\opt}{opt}
\DeclareMathOperator{\sol}{sol}
\DeclareMathOperator{\area}{area}
\DeclareMathOperator{\poly}{poly}

\newtheorem{definition}{Definition}
\newtheorem{theorem}{Theorem}
\newtheorem{lemma}{Lemma}
\newtheorem{corollary}{Corollary}
\newtheorem{remark}{Remark}

\makeatletter
\newtheorem*{rep@theorem}{\rep@title}
\newcommand{\newreptheorem}[2]{%
	\newenvironment{rep#1}[1]{%
		\def\rep@title{#2 \ref{##1}}%
		\begin{rep@theorem}}%
		{\end{rep@theorem}}}
\makeatother
\newreptheorem{theorem}{Theorem}

\makeatletter
\newcommand{\algmargin}{\the\ALG@thistlm}
\makeatother
\newlength{\whilewidth}
\algdef{SE}[parWHILE]{parWhile}{EndparWhile}[1]
{\parbox[t]{\dimexpr\linewidth-\algmargin}{%
		\hangindent\whilewidth\strut\algorithmicwhile\ #1\ \algorithmicdo\strut}}{\algorithmicend\ \algorithmicwhile}%
\algnewcommand{\parState}[1]{\State%
	\parbox[t]{\dimexpr\linewidth-\algmargin}{\strut #1\strut}}

\journal{Information Processing Letters}

\begin{document}

\begin{frontmatter}

\title{Minimum Enclosing Rectangle with Outliers}

\author{Zhengyang Guo\corref{cor1}}
\ead{GUOZ0015@e.ntu.edu.sg}

\author{Yi Li}
\ead{yili@ntu.edu.sg}

\address{School of Physical and Mathematical Sciences, Nanyang Technological University, Singapore}
           
\cortext[cor1]{Corresponding author}

\begin{abstract}
We study the problem of minimum enclosing rectangle with outliers, which asks to find, for a given set of $n$ planar points, a rectangle with minimum area that encloses at least $(n-t)$ points. The uncovered points are regarded as outliers. We present an exact algorithm with $O(kt^3+ktn+n^2\log n)$ runtime, assuming that no three points lie on the same line. Here $k$ denotes the number of points on the first $(t+1)$ convex layers. We further propose a sampling algorithm with runtime $O(n+\mbox{poly}(\log{n}, t, 1/\epsilon))$, which with high probability finds a rectangle covering at least $(1-\epsilon)(n-t)$ points with at most the exact optimal area.
\end{abstract}


\end{frontmatter}


\section{Introduction}\label{sec:intro}
In this work, we consider the problem of partial minimum enclosing rectangle, which is a generalization of classic minimum enclosing rectangle to the cases where there are outliers. Outliers have attracted increasing attention in the studies of computational geometry, and a recent trend is to combine shape fitting tasks (see \cite{har2004shape} for what these task are) and outlier detection. Examples include the projective clustering \cite{mishra2005sublinear}, unit disk cover \cite{gandhi2004approximation,ghasemalizadeh2012improved}, $k$-center/means/median clustering \cite{malkomes2015fast,guha2017distributed}, subspace approximation \cite{deshpande2020subspace}, minimum enclosing ball \cite{ding2018solving} and subspace clustering \cite{soltanolkotabi2012geometric}. As the geometric shapes are sensitive to outliers, removing outliers can sometimes greatly improve the quality of the output. 

The classic minimum enclosing rectangle (\MER) or parallelogram (\MEP) is of interest in digital signal processing \cite{schwarz1994finding,o1985finding} and computer graphics \cite{chen1997determining,barequet2001efficiently,chan2001determination}. It works as a preprocessing to obtain a bounding box for the input point cloud \cite{o1985finding}. For 2D points, the objective is to find a rectangle \cite{das2005smallest} or parallelogram \cite{schwarz1994finding} of minimum area that circumscribes all or most of the points. For 3D points, the aim is to find a hyper-rectangle \cite{mount2000quantile} or parallelepiped \cite{vivien2004minimal} of minimum volume.

In \cite{kaplan2019finding}, Kaplan et al.\ consider a relatively restricted case where the rectangle is axis parallel and gives an $O(n^{5/2}\log^2{n})$ exact algorithm. In \cite{schwarz1994finding}, Schwarz gives an exact algorithm to find a parallelogram of minimum area that encloses a convex polygon. The algorithm is linear with respect to the number of vertices of the convex polygon. Together with the $O(n\log{N})$-time algorithms in \cite{kirkpatrick1986ultimate} or \cite{chan1996optimal}, where $N$ denotes the convex hull size, we can find the \MEP\ of $n$ planar points in $O(n\log N)$ time by finding the convex hull first and then its minimum enclosing parallelogram.

For the outlier cases, there was an $O(n+t^2n)$-time algorithms for \MER~\cite{segal1998enclosing}, under the assumption that the rectangle is axis-parallel. The runtime was later improved to $O(n+t^3\log^{2}{n})$ for the cases $t < n/\log^2 n$~ \cite{mahapatra2011k}. Finding \MER\ of arbitrary orientation leads to a significantly higher runtime of $O(n^{2}t^{2}+n^{2}t \log{n})$ for $t<n/2$~\cite{das2005smallest}.

\subsection{Notations and Problem Formulation}\label{formulation}
Let $X$ be the input set of $n$ planar points and $t$ be the number of outliers. Besides, $N$ denotes the convex hull size of $X$ and $k=k(X,t)$ denotes the number of points on the first $(t+1)$ convex layers (see \cite{dalal2004counting} for the definition of convex layer, we also restate it in Appendix~\ref{sec::appendix}). We assume $t<n/2$ as in general there are more inliers than outliers. The main problem is defined below.

\begin{definition}[Minimum Enclosing Rectangle (Parallelogram) with Outliers]
Given a set $X$ of $n$ planar points and an integer parameter $t$, the task is to find a rectangle (parallelogram) with the minimum area that covers at least $(n-t)$ points.
\label{parallelogram formulation}
\end{definition}

We use $\MER(X,t)$ to denote the problem itself, $\sol^\ast(X,t)$ to denote the optimal solution and $\opt(X,t)$ to denote the optimal area. Besides, we use uppercase letters $X$, $L$, $R$, $J$, $H$, $G$, $M$, $\Lambda$ and $\Gamma$ to indicate a collection of objects (such as sets, lists, arrays) in this paper. For a list $L$, we use $L_i$ to denote its $i$-th element. For three distinct points $O,A,B$, we use $\angle{AOB}$ (in radians) to denote the clockwise angle from the \emph{ray} $OA$ to the \emph{ray} $OB$.

\begin{figure}[!t]
	\begin{minipage}{6.1cm}
	\centering
	\begin{tikzpicture}[scale=0.5]
	\def\centerarc[#1](#2)(#3:#4:#5)
	{ \draw[#1] ($(#2)+({#5*cos(#3)},{#5*sin(#3)})$) arc (#3:#4:#5); }
	
	\draw[dashed, very thick] (0,0) -- (6.6,0) -- (6.6, 4) -- (0, 4) -- cycle;
	
	\draw node[fill,circle,minimum width=4pt, inner sep=0pt, label={[label distance=0.01cm]270:$A_8$}](A8) at (2.7, 0) {};
	\draw node[fill,circle,minimum width=4pt, inner sep=0pt, label={[label distance=0.2cm]270:$A_1$}](A1) at (4.5, 0.3) {};
	\draw node[fill,circle,minimum width=4pt, inner sep=0pt, label={[label distance=0cm]0:$A_2$}](A2) at (6, 1.2) {};
	\draw node[fill,circle,minimum width=4pt, inner sep=0pt, label={[label distance=0cm]0:$A_3$}](A3) at (6.6, 3) {};
	\draw node[fill,circle,minimum width=4pt, inner sep=0pt, label={[label distance=0.01cm]90:$A_4$}](A4) at (5.1, 4) {};
	\draw node[fill,circle,minimum width=4pt, inner sep=0pt, label={[label distance=0.01cm]90:$A_5$}](A5) at (2, 4) {};
	\draw node[fill,circle,minimum width=4pt, inner sep=0pt, label={[label distance=0cm]180:$A_6$}](A6) at (0,2.7) {};
	\draw node[fill,circle,minimum width=4pt, inner sep=0pt, label={[label distance=-0.3cm]210:$A_7$}](A7) at (0.5,0.9) {};
	
	\draw node[fill,circle,minimum width=4pt, inner sep=0pt] at (3.2, 3.2) {};
	\draw node[fill,circle,minimum width=4pt, inner sep=0pt] at (3.7, 2) {};
	\draw node[fill,circle,minimum width=4pt, inner sep=0pt] at (5, 2.8) {};
	\draw node[fill,circle,minimum width=4pt, inner sep=0pt] at (1.8, 2.3) {};
	\draw node[fill,circle,minimum width=4pt, inner sep=0pt] at (1.3, 1.3) {};
	\draw node[fill,circle,minimum width=4pt, inner sep=0pt] at (2.9, 0.7) {};
	\draw node[fill,circle,minimum width=4pt, inner sep=0pt] at (5, 1.1) {};
	
	\draw[very thick] (A8)--(A1)--(A2)--(A3)--(A4)--(A5)--(A6)--(A7)--(A8);
	
	\end{tikzpicture}
	\caption{Relative position of the \MER\ and the convex hull of a point set. The vertex hull is $A_1A_2\cdots A_8$. The top side of the \MER\ contains $A_4A_5$ and the bottom side contains $A_8$. The left side contains $A_6$ and the right side $A_3$.}
	\label{fig:collinear}
	\end{minipage}
	\hfill
	\begin{minipage}{5.5cm}
		\centering
		\begin{tikzpicture}[scale=0.65]
		\makeatletter 
		\pgfdeclarepatternformonly[\LineSpace,\tikz@pattern@color]{my north east lines}{\pgfqpoint{-1pt}{-1pt}}{\pgfqpoint{\LineSpace}{\LineSpace}}{\pgfqpoint{\LineSpace}{\LineSpace}}%
		{
			\pgfsetcolor{\tikz@pattern@color} 
			\pgfsetlinewidth{0.4pt}
			\pgfpathmoveto{\pgfqpoint{0pt}{0pt}}
			\pgfpathlineto{\pgfqpoint{\LineSpace + 0.1pt}{\LineSpace + 0.1pt}}
			\pgfusepath{stroke}
		}
		\makeatother 
		\newdimen\LineSpace
		\tikzset{
			line space/.code={\LineSpace=#1},
			line space=3pt
		}
	
		\fill[gray,pattern=my north east lines,line space = 8pt,dashed] (-0.5,-0.3) -- (5,-0.3) -- (5,5) -- (-0.5,5) -- cycle;
		
		\draw [fill,white] (1.45, -0.3) -- (9/17, 115/34)--(5, 4.5) -- (5,-0.3) -- cycle;
		\draw [thick] (1.45, -0.3) -- (9/17, 115/34) -- (5,4.5);
		
		\draw node[fill,circle,minimum width=4pt, inner sep=0pt, label={[label distance=0.03cm]280:$P_1$}] at (1, 3.5) {};
		\draw node[fill,circle,minimum width=4pt, inner sep=0pt, label={[label distance=0.05cm]270:$P_2$}] at (3, 4) {};
		\draw node[fill,circle,minimum width=4pt, inner sep=0pt, label={[label distance=0.05cm]0:$P_4$}] at (1, 1.5) {};
		
		\draw node at (3.5,1.5) {Enclosed};
		\draw node at (1.7,4.3) {Excluded};
		
		\end{tikzpicture}
		\caption{Enclosed and excluded areas. The boundaries are made up by three parts: the line $P_1P_2$ and the perpendicular line through $P_4$. The solid line boundaries are considered a part of the enclosed area.}
		\label{fig:sides}
	\end{minipage}
\end{figure}

\subsection{Our Contributions}
We give an exact algorithm to $\MER(X,t)$.  When $t<n/2$, its time complexity is $O(nt^3+n^2t+n^2\log{n})$,  better than $O(n^{2}t^{2}+n^2t\log{n})$ in~\cite{das2005smallest}. Both complexities are computed for the cases where there are no three collinear points. The major difference between our work and~\cite{das2005smallest} is that we define the notion of \emph{valid pairs} and prove there are at most $O(nt)$ of them. Using valid pairs, we can locate one side of the rectangle.
For recent results on the minimum enclosing rectangle, we refer the readers to Table~\ref{tab:summary}.

\begin{theorem}
	\label{thm:main1}
	Given a set of $n$ points, we can find its \MER\ with $t$ outliers in time  $O(kt^3+ktn +n^{2}\log{n})$, where $k=k(X,t)$ is the number of points on the first $(t+1)$ convex layers..
\end{theorem}

Though $k$ can be as large as $n$ in the worst case, in general it is much smaller. See Remark~\ref{worst v.s. average} for a discussion of this. When $k, t\ll n$, the $n^{2}\log{n}$ term becomes dominant in the time complexity and the algorithm is not adequately efficient for implementation. We therefore sample the point set uniformly at random and show an $O(n+\poly(\log{n}, t, 1/\epsilon))$-time approximation algorithm.

\begin{theorem}
	\label{thm:main2}
	There is a sampling algorithm which, given a set $X$ of $n$ points, with probability at least $(1-3/n)$, finds a rectangle of area at most $\opt(X,t)$ such that the number of enclosed points is between $(n-t)(1-\epsilon)$ and $(n-t+1)$, in time $O(n+\poly(\log{n}, t, 1/\epsilon))$.
\end{theorem}

\begin{table}[!t]
	\footnotesize
	\centering
	\begin{tabular}{|c|c|c|c|}
		\hline
		Orientation & Paper & Restriction & Time Complexity\\
		\hline
		\multirow{5}{*}{axis-parallel} & \cite{mahapatra2011k} & $t<n/2$ & $O\left(n+t^3\log^2{t}\right)$\\
		\cline{2-4}
		& \cite{mahapatra2011k} & -- & $O\left(n^3\log^2{n}\right)$\\
		\cline{2-4}
		& \cite{mahapatra2012smallest} & -- & $O\left(n^4\right)$\\
		\cline{2-4}
		& \cite{chan2020smallest} & -- & $O\left(n^2\log{n}\right)$\\
		\cline{2-4}
		& \cite{de2016covering} & -- & $O\left(n(n-t)^2\log{n}+n\log^2{n}\right)$\\
		\hline
		\multirow{3}{*}{arbitrary} & \cite{chan2020smallest} & -- & $O\left(n^3\log{n} + n^3(n-t)/2^{\Omega(\sqrt{\log(n-t)})}\right)$\\
		\cline{2-4}
		& \cite{das2005smallest} & -- & $O\left(n^2\log{n}+nt^2(n-t)+nt(n-t)\log(n-t)\right)$\\
		\cline{2-4}
		& this work & $t<n/2$ & $O\left(n^2\log{n}+n^2t+nt^3\right)$\\
		\hline
	\end{tabular}
	\caption{}
	\label{tab:summary}
\end{table}

\section{Preliminaries}
The observation is a result from \cite{toussaint1983solving}, which concerns the relative position between the \MER\ and the convex hull. See Figure~\ref{fig:collinear} for illustration. 

\begin{lemma}[{\cite[Theorem 2.1]{toussaint1983solving}}]
\label{collinear rectangle}
For any given planar point set $X$ and its conex hull $\CH(X)$, there exists a \MER\ such that one of its four sides must contain a side of $\CH(X)$, and each of the other side must pass a vertex of $\CH(X)$.
\end{lemma}

\section{Algorithm}
\label{sec:timealg}
We shall analyze what properties the optimal solution must satisfy and build the algorithm along the way. We will ignore the cases where there are collinear points in the discussion of the geometric properties of our concerned problem. It is common to ignore the degenerate or corner cases, for instance, such practice was thoroughly adopted in \cite{van2000computational}. For completeness, we shall discuss in Appendix~\ref{sec:three_points} how to modify the algorithm to remove the assumption. Below is a direct generalization of Lemma~\ref{collinear rectangle} to the outlier case. 

\begin{corollary}
	\label{collinear parallelogram with outliers}
	The optimal solution $\sol^\ast(X,t)$ and the convex hull of the $(n-t)$ enclosed points have the same positional relation as described in Lemma~\ref{collinear rectangle}. 
\end{corollary}

For further discussion, we define the enclosed and excluded areas. We also illustrate the notion in Figure~\ref{fig:sides}.
\begin{definition}[Enclosed/Excluded Area]
	\label{def:minority-majority}
	For the open areas divided by a line or a simple polyline, the one containing more points is called the \emph{enclosed area}, the union of the others is called the \emph{excluded area}.
\end{definition}

We let $P_1P_2$ denote the side of the convex hull that is covered by one side of $\sol^\ast(X,t)$, $P_3$ denote the point that is passed through by the opposite side, $P_4$ and $P_5$ denote the points that are on the other two sides respectively. See Figure~\ref{fig:sides} for illustration. For the number of enclosed points, we have the following lemma.
\begin{lemma}
\label{lem:enclosed points number}
$\opt(X,t)$ is either $(n-t)$ or $(n-t+1)$.
\end{lemma}

\begin{proof}
	By Definition~\ref{parallelogram formulation}, we have $\opt(X,t)\geq n-t$. If more than $(n-t+1)$ points are enclosed, then moving a side inwards will exclude at most 2 points (recall we assume there are no three collinear points), resulting in a smaller rectangle enclosing at least \((n-t)\) points, contradicting the optimality of $\sol^\ast(X,t)$.
\end{proof}

While the optimal rectangle may enclose $(n-t)$ or $(n-t+1)$ points, for simplicity, we will only discuss the situation where the rectangle encloses $(n-t)$ points in Section~\ref{sec:p1} and~\ref{sec:p2}. The analysis and algorithms can be applied to the case of $(n-t+1)$ points, simply by replacing $t$ with $(t-1)$ in $(n-t)$. As there are at least $(n-t)$ points (including $P_{1}$, $P_{2}$ and $P_{3}$) between the \emph{line} $P_{1}P_{2}$ and its parallel through $P_{3}$, at most $t$ points therefore lie in the excluded area of the line $P_1P_2$. We next give the definition of valid pair.

\begin{definition}[Valid Pair]
\label{valid}
Given a pair of points $(P_{1}, P_{2})$, if the excluded area of the \emph{line} $P_1P_2$ contains at most $t$ points, then $(P_{1}, P_{2})$ is called a valid pair.
\end{definition}

The points $P_1$, $P_2$, $P_3$, $P_4$ and $P_5$ which are on the sides of a rectangle that encloses $(n-t)$ points must satisfy the following conditions. First, $(P_1, P_2)$ must be a valid pair. Second, there are at least $(n-t)$ points (including $P_{1}$, $P_{2}$ and $P_{3}$) between the \emph{line} $P_{1}P_{2}$ and its parallel through $P_3$. Third, $P_4$ and $P_5$ must lie between the \emph{line} $P_1P_2$ and its parallel through $P_3$.

After finding the five points, we can then determine the four sides of the rectangle. We draw one line passing through $P_1$ and $P_2$, and then a parallel through $P_3$, and two perpendicular lines through $P_4$ and $P_5$ respectively. The enclosed area of the four lines is the desired rectangle.

\begin{algorithm}[!t]
	\caption{Searching for all valid pairs $(P_1,P_2)$}
	\label{alg:valid_pairs}
	\fontsize{7.5}{8}\selectfont
	\begin{algorithmic}[1]
		\Function{Valid-Pairs}{$X,t$}
		\State $\Lambda \gets \emptyset$
		\Comment{the set is to collect all the valid pairs}
		\For{$P_1$ in $X$}
		\State $P_0\gets$ arbitrary point in $X\setminus \{P_1\}$
		\Comment{starting point of the rotation}
		\State $L, R\gets $ points in $X\setminus \{P_0,P_1\}$ on the left and right side of \emph{ray} ${P_0P_1}$
		\label{alg:valid_pairs:LR}
		\State Sort $L$ and $R$ respectvely in increasing order of the clockwise angle around $P_0$ \label{alg:valid_pairs:LRsort}
		\State $n_L\gets|L|$, $n_R\gets |R|$
		\Comment{$n_L$ and $n_R$ record the size of $L$ and $R$}
		\While{the clockwise angle $P_1P_0P_2 < \pi$}
		\label{alg:valid_pairs:loop_content_start}
		\If{$\angle L_1P_0R_1<\pi$}
		\Comment{determine whether \emph{line} $P_1P$ meets $R_1$ or $L_1$ next}
		\State $P_2\gets R_1$
		\State  pop $R_1$ from $R$ and add it to the end of $L$ \label{alg:valid_pairs:H1}
		\State $n_R\gets n_R-1$, $n_L\gets n_L+1$
		\Else
		\State $P_2\gets L_1$
		\State pop $L_1$ from $L$ and add it to the end of $R$ \label{alg:valid_pairs:H2}
		\State $n_R\gets n_R+1$, $n_L\gets n_L-1$
		\EndIf
		\If{$n_L\leq t$}
		\State add $(P_1, P_2, n_L,R)$ into $\Lambda$
		\EndIf
		\If{$n_R\leq t$}
		\State add $(P_1, P_2, n_R,L)$ into $\Lambda$
		\EndIf
		\EndWhile
		\label{alg:valid_pairs:loop_content_end}
		\EndFor
		\State \Return $\bm\Lambda$
		\EndFunction
	\end{algorithmic}
\end{algorithm}

\subsection{Searching for $P_1$ and $P_2$}
\label{sec:p1}
We now discuss how to find all the valid pairs $(P_1, P_2)$. The technique we used here is known as rotating calipers \cite{toussaint1983solving}, a powerful tool widely employed in solving problems of computational geometry. For a given point $P_1$ and an arbitrary but fixed point $P_0\in X\setminus\{P_1\}$, we divide the plane into two open halves. The open half plane on the clockwise side of the \emph{ray} $P_1P_0$ is referred to as its right side, the other its left side. We then sort the points on the left and the right sides respectively in the clockwise order around $P_1$. The left and right sorted lists are denoted by $L$ and $R$ accordingly.

We start $P$ from $P_0$ and rotate \emph{ray} $P_1P$ around $P_1$ for a half circle. During the rotation, we use $P_2$ to keep the latest point met by \emph{line} $P_1P$. The lists $L$ and $R$ are also updated to store points on the left and the right sides of the rotating \emph{ray} $P_1P$ respectively. The clockwise angle $\angle L_1P_1R_1$ indicates whether $L_1$ or $R_1$ will be first hit by the rotating \emph{line} $P_1P$ and thus which side the next $P_2$ should come from. If $\angle L_1P_1R_1<\pi$, the next $P_2$ will come from the right side, and else if $\angle L_1P_1R_1>\pi$, the left side. Note that under the assumption that no three points lie on the same line, $\angle L_1P_1R_1$ can never be $\pi$. The point $P_2$ will be popped from the head of the list and then added to the end of the other list on the opposite side.  The size of the first list will decrease by $1$ and that of the second will increase by $1$. We use two auxiliary variables $n_L$ and $n_R$ to record the sizes of $L$ and $R$ respectively. When $n_L\leq t$ or $n_R\leq t$, we can conclude that $(P_1, P_2)$ is a valid pair.

Iterating the process over $P_1\in X$, we can find all the valid pairs $(P_1, P_2)$ and store them in set $\Lambda$.  For further convenience, we also record the set of points in the enclosed area. The algorithm is presented in Algorithm~\ref{alg:valid_pairs}.  We also illustrate the rotation of \emph{ray} $P_1P$ by an example in Figure~\ref{fig:rotation}.

\begin{figure}[!t]
	\begin{minipage}[c]{0.4\textwidth}
		\centering
		\begin{tikzpicture}[scale=0.65]
		\def\centerarc[#1](#2)(#3:#4:#5)
		{ \draw[#1] ($(#2)+({#5*cos(#3)},{#5*sin(#3)})$) arc (#3:#4:#5); }
		
		\node [label={[shift={(-0.3,-0.3)}]$P_1$}, fill, draw, circle, minimum width=3pt, inner sep=0pt] at (0,0) {};
		
		\draw[dashed] (0,0) -- (2,2.5);
		\draw[-latex,thick] (0,0) -- (-2,-2.5);
		\node at (-2,-1.4) {$A_1(P_2)$};
		\node at (-1.7,-0.5) {$A_2$};
		\node at (-2.3,0.5) {$A_3$};
		\node at (-1.5,1.7) {$A_4$};
		\node at (0,2) {$A_5$};
		\node at (2.1,1.2) {$A_6$};
		\node at (2.4,-0.9) {$A_7$};
		\node at (0.8,-1.5) {$A_8$};
		\node at (0.2,-2.4) {$A_9$};
		
		\centerarc[-latex,thick,line cap=round](0,0)(227:198:3.5);
		\centerarc[-latex,thick,line cap=round](0,0)(47:18:3.5);
		
		\draw node[fill,circle,minimum width=4pt, inner sep=0pt] at (-1.2,1.3) {};
		\draw node[fill,circle,minimum width=4pt, inner sep=0pt] at (-2,0.7) {};
		\draw node[fill,circle,minimum width=4pt, inner sep=0pt] at (-1.4,-0.3) {};
		\draw node[fill,circle,minimum width=4pt, inner sep=0pt] at (0,1.7) {};
		\draw node[fill,circle,minimum width=4pt, inner sep=0pt] at (1.7,1.3) {};
		\draw node[fill,circle,minimum width=4pt, inner sep=0pt] at (2,-0.9) {};
		\draw node[fill,circle,minimum width=4pt, inner sep=0pt] at (0.5,-1.3) {};
		\draw node[fill,circle,minimum width=4pt, inner sep=0pt] at (0.2,-2) {};
		\draw node[fill,circle,minimum width=4pt, inner sep=0pt] at (-1.12,-1.4) {};
		\end{tikzpicture}
	\end{minipage}
	\hfill
	\begin{minipage}[c]{0.5\textwidth}
		\caption{Clockwise rotation of the \emph{ray} $P_1P$. In the current position, $P_2=A_1$, the left side is $(A_6,A_7,A_8,A_9)$ and the right side is $(A_2,A_3,A_4,A_5)$. Since $\angle A_6P_1A_2>\pi$, the next $P_2$ is $A_6$. The new left side will become $(A_7,A_8,A_9,A_1)$  and the new right side will be as $(A_2,A_3,A_4,A_5,A_6)$.}
		\label{fig:rotation}
	\end{minipage}
\end{figure}
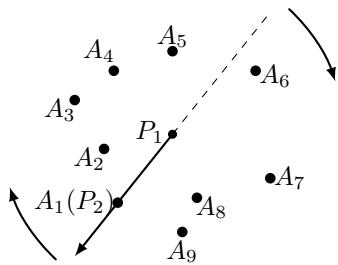

\begin{algorithm}[!t]
	\caption{Searching for $P_3$, $P_4$ and $P_5$}
	\label{alg:para}
	\fontsize{7.5}{8}\selectfont
	\begin{algorithmic}[1]
		\Function{Enclose}{$M, t, P_1, P_2, m$}
		\Comment{$m$ is the number of points in the excluded area of \emph{line} $P_1P_2$, $M$ is the set of points in the enclosed area of \emph{line} $P_1P_2$}
		\State $J\gets$ the $(t+1)$ leftmost points in $M$ along the direction of \emph{ray} $P_1P_2$
		\State $H\gets$ the $(t+1)$ rightmost points in $M$ along the direction of \emph{ray} $P_1P_2$
		\State $G\gets$ the $(t+1)$ farthest points in $M$ from \emph{line} $P_1P_2$
		\State store $J$ and $H$ in two linked lists, both are ordered from left to right
		\State store $G$ in a linked list, whose point are in the decreasing order of their distance to $P_1P_2$
		\State $\Gamma\gets\emptyset$
		\For{$i \gets 1, 2,\dots, t+1$}
		\State $P_4\gets J_i$
		\State $x\gets$ the number of points excluded by \emph{line} $P_1P_2$ and the perpendicular line through $P_4$
		\State $c\gets t+1$
		\Comment{$c$ stores the number of excluded points}
		\State $k\gets 1$
		\Comment{initialize the index of the candidate for $P_3$}
		\label{alg:para:finding_start}
		\For{$j \gets x+1,\cdots,t+1$}
		\State $P_5\gets H_j$
		\If{$P_5$ lies between the \emph{line} $P_1P_2$ and the parallel through $P_3$}
		\If{$j < t+1$ \textbf{and} ($H_{j+1}$ is between \emph{line} $P_1P_2$ and the parallel through $P_3$) \textbf{and} \newline \indent\indent\indent\ \ \ \ ($H_jH_{j+1}$ is perpendicular to $P_1P_2$)}
		\State $c\gets c-2$
		\Else
		\State $c\gets c-1$
		\EndIf
		\EndIf
		\While{$c < t$ \textbf{and} $k<t+1$}
		\If{$G_k$ is between the perpendicular lines through $P_4$ and $P_5$}
		\If{$G_kG_{k+1}$ is parallel to $P_1P_2$}
		\State $c\gets c+2$
		\State $k\gets k+2$
		\Else
		\State $c\gets c+1$
		\State $k\gets k+1$
		\EndIf
		\EndIf
		\EndWhile
		\If{$c=t$}
		\State $P_3\gets G_k$
		\State add $(P_3, P_4, P_5)$ to $\Gamma$
		\EndIf
		\label{alg:para:finding_end}
		\EndFor
		\EndFor
		\State \Return $\Gamma$
		\EndFunction
	\end{algorithmic}
\end{algorithm}

\subsection{Finding $P_3$, $P_4$ and $P_5$}
\label{sec:p2}

First, $P_3$, $P_4$ and $P_5$ must be in the enclosed area of the \emph{line} $P_1P_2$. Without loss of generality, we may assume $P_4$ is on the left of $P_5$ along the direction of the \emph{ray} $P_1P_2$. Furthermore, $P_3$ is among the $(t+1)$ farthest points from the \emph{line} $P_1P_2$, $P_4$ is among the $(t+1)$ leftmost points and $P_5$ is among the $(t+1)$ rightmost points. We can use the well-known min or max heap to find all the candidates for $P_3$, $P_4$ and $P_5$. For further convenience, the candidates for $P_3$ are then stored in a linked list, in the decreasing order of their distance to the \emph{line} $P_1P_2$. So are the candidates for $P_4$ and $P_5$, both of which are ordered from left to right.

When $P_1$, $P_2$ and $P_4$ are given,  there can be at most $(t+1)$ pairs of $(P_3, P_5)$ such that the rectangle encloses $(n-t)$ points. We let $x\leq t$ denote the number of points excluded by \emph{line} $P_1P_2$ and the perpendicular line through $P_4$. We show how to find all the pairs $(P_3, P_5)$. We claim that we can skip the first $x$ candidates for $P_5$. Indeed, if we choose any of the first $x$ candidates as $P_5$, there will be at least $(t+1-x)$ another excluded points, resulting in at least $(t+1)$ excluded points in total. We therefore initialize $P_5$ as the $(x+1)$-th candidate, and $P_3$ as the $1$st candidate. The total number of excluded points is now $t$. Whenever we go to the next candidate for $P_5$, we need to check whether it is between the \emph{line} $P_1P_2$ and the parallel through $P_3$. If it is, then the number of excluded points will decrease by $1$ or $2$, depending on whether $1$ or $2$ points will be included when moving the perpendicular line outwards. If not, the number does not change. We continue moving $P_3$ to the next candidate until the number of excluded points becomes at least $t$. We repeat moving $P_5$ and $P_3$ this way. The whole process is presented in Algorithm~\ref{alg:para}.

\subsection{Finding \MER}
We finally enumerate all possible combinations of $P_1$, $P_2$, $P_3$, $P_4$ and $P_5$, and find the one of the minimum area among all the rectangles which encloses $(n-t)$ or $(n-t+1)$ points. The overall algorithm is presented in Algorithm~\ref{alg:main}. Notice that the rectangle with overlapping points on its sides is merely a special case of the ones where the six points do not overlap. The algorithm can be adapted to the overlapping cases and this will only increase a constant factor to the overall time complexity.

\begin{algorithm}[!t]
	\caption{Partial-\MER}
	\label{alg:main}
	\fontsize{7.5}{8}\selectfont
	\begin{algorithmic}[1]
		\Require a planar point set $X$, the number of outliers $t$.
		\State $A\gets +\infty$
		\State $rec^\ast\gets nil$
		\Comment{Initialize the area of the rectangle}
		\For{$(P_{1}, P_{2}, m, M)$ in $\Call{valid-pairs}{X,t}$} \label{alg:main:loop_begin}
		\Comment{Algorithm~\ref{alg:valid_pairs}}
		\For{$(P_3, P_4, P_5)$ in $\Call{enclose}{M, t, P_1, P_2, m}\cup\Call{enclose}{M, t-1, P_1, P_2, m}$}
		\Comment{Algorithm~\ref{alg:para}}
		\parState{$rec\gets$ the rectangle with $P_1$ and $P_2$ on one side, $P_3$ on the opposite side, $P_4$ and $P_5$ on the two perpendicular sides respectively}
		\If{$\area(rec) < A$}
		\State $rec^\ast \gets rec$
		\State $A\gets \area(rec)$ \label{alg:main:loop_end}
		\EndIf
		\EndFor
		\EndFor
		\State \Return $para^\ast$
	\end{algorithmic}
\end{algorithm}

\section{Time Analysis of Algorithm~\ref{alg:main}}
\label{sec:analysis}
In order to enclose $(n-t)$ or $(n-t+1)$ points, there must be at most $t$ points in the excluded area of the \emph{line} $P_1P_2$. In another word, $(P_1, P_2)$ is a valid pair. For a fixed point $P_1$, we let \emph{ray} $P_1P$ rotate a half circle around $P_1$ in the clockwise order. In the rotation, the points met by \emph{line} $P_1P$ in order are marked by $X_1,X_2,\dots, X_{n-1}$ . Let $n_L^{(i)}$ and $n_R^{(i)}$ be the number of points on the left and on the right side of the \emph{ray} $P_1P$ respectively, when the \emph{line} $P_1P$ meets $X_i$ during the rotation. The following lemma tracks how $n_L^{(i)}$ and $n_R^{(i)}$ change in the process.

\begin{lemma}\label{change}
Rotate the \emph{ray} $P_1P$ in the clockwise order. The \emph{line} $P_1P$ will meet the points of $X\setminus\{P_{1}\}$ in the sequence of $X_1,X_2,\dots, X_{n-1}$. Let $n_{L}^{(i)}$ and $n_{R}^{(i)}$ be the number of points on the left and right side of the \emph{ray} $P_{1}P$ respectively, when the \emph{line} $P_1P$ meets the point $X_i$. Then
\[
n_{L}^{(i)} =  n_{L}^{(i-1)} + \delta_i \quad\text{and}\quad n_{R}^{(i)} =  n_{R}^{(i-1)} - \delta_i,
\]
where
\[
\delta_i = \begin{cases}
			1, & \text{if } X_i \text{ is previously on the right side of the \emph{ray} }P_1P;\\
			-1, & \text{if } X_i \text{ is previously on the left side of the \emph{ray} }P_1P;
		\end{cases}
\]

\end{lemma}

Suppose there is a point $P_0\neq P_1$ such that there are $s$ points on the right side of \emph{ray} $P_1P_0$, where $s\leq t$. We initialize $P$ to be $P_{0}$ and rotate $P$ around $P_1$ by 180 degree. The \emph{line} $P_1P$ sweeps over all the points except $P_1$. And the initial left side of \emph{ray} $P_1P$ becomes the right side. This indicates that there are finally $(n-s-2)$ points on the right side, that is, $n_R^{(n-1)} = n-s-2$. We are now ready to prove an auxiliary lemma, by which we further prove that  there are not too many valid pairs.
\begin{lemma}
Suppose a sequence of numbers $x_1,x_2,\dots,x_{n-1}$ satisfies that $x_{1}=s$, $x_{n-1}=n-s-2$, $x_{i+1}-x_{i}\in \{-1,1\}$. When $s\leq t$, there are at most $(2t+1)$ indices $i$ satisfying $x_{i}\leq t$.
\end{lemma}

\begin{proof}
We let $i^\ast$ denote the largest index $i$ such that $x_{i}\leq t$. Then we have
$x_{n-1}-x_{i^\ast}\geq n-s-2-t\geq n-2t-2$. On the other hand, $x_{n-1}-x_{i^\ast}=\sum_{i=i^\ast}^{n-2}(x_{i+1}-x_i)\leq \sum_{i=i^\ast}^{n-2}1=n-i^\ast-1$. By the two inequalities, we have $i^\ast\leq 2t+1$ and the possible $i$ satisfying $x_{i}\leq t$ can only be in $i\in\{1,2,...,2t+1\}$.
\end{proof}

The bound can be applied to both the left and the right sides of \emph{ray} $P_{1}P$. Therefore we have the following corollary.
\begin{corollary}\label{cor:pair_count_fixed_P1}
Given a fixed $P_{1}$, there are at most $2\cdot(2t+1)=4t+2$ valid pairs $(P_{1},P_{2})$.
\end{corollary}
Besides, we have the following observation on the position of $P_{1}$.

\begin{lemma}
$P_{1}$ can only be on the first $(t+1)$ convex layers of \(X\).
\end{lemma}

\begin{proof}
If $P_{1}$ is inside the $(t+1)$-th convex layer, then any line $\ell$ through $P_{1}$ must intersect each of the first $(t+1)$ convex layers. Therefore, on any side of $\ell$, there must be at least one point from each of the first $(t+1)$ convex layers. In total, there would be at least $(t+1)$ points, which contradicts the fact that $(P_1,P_2)$ is a valid pair. We therefore conclude that $P_{1}$ must be on the first $(t+1)$ convex layers.
\end{proof}
We can now bound the total number of valid pairs.

\begin{corollary}
There are in total $O(tk)$ valid pairs $(P_{1}, P_{2})$.
\label{number}
\end{corollary}

\begin{proof}
Each valid pair $(P_{1},P_{2})$ is counted twice when $P_{1}$ is fixed and when $P_{2}$ is fixed. Thus there are $\frac{1}{2}\cdot (4t+2)\cdot k=(2t+1)k$
valid pairs in total.
\end{proof}

\begin{remark}\label{worst v.s. average}
	Since $k\leq n$, the estimate in the preceding corollary is $O(nt)$. We show that this bound cannot be substantially improved. Let $A_1,\dots,A_n$ be the $n$ vertices of a convex polygon in clockwise order, then $k=n$ in this case. Note that if the number of points strictly on one side of \emph{line} $A_{i}A_{j}$ is no more than $t$,  there would be at most $t$ other vertices between $A_{i}$ and $A_{j}$. The total number of valid pairs $\{A_{i}, A_{j}\}$ is exactly $n(t+1)=\Theta(kt)=\Theta(nt)$.
	
	However, the situation above where there is only one single convex layer is very rare. Usually, $k$ is much smaller than $n$. It is known that the expected convex hull size of $n$ points sampled from a general convex body in the plane is $O(n^{2/3})$~\cite{buchta2012boundary}. A direct corollary is that $\mathbb{E}k = O(t n^{2/3})$, much smaller than $n$ when $t\ll n^{1/3}$. Another example of $k\ll n$ is when $n$ points are randomly sampled from a component independent distribution on the plane. Under such conditions, the expected size of the $i$-th convex layer is proved to be $O(i^{2}\log(n/i))$~\cite{he2018maximal}. Consequently, $\mathbb{E}k = O(t^{3}\log(n/t))$. 
\end{remark}

The following observations are direct and simple. First, $P_3$, $P_4$ and $P_5$ are in the enclosed area of \emph{line} $P_1P_2$. Second, $P_3$ is among the $(t+1)$ farthest points from \emph{line} $P_1P_2$. Third, $P_4$ is among the $(t+1)$ leftmost points along the direction of \emph{ray} $P_1P_2$. Lastly, $P_5$ is among the $(t+1)$ rightmost points along the direction of \emph{ray} $P_1P_2$.
Now we prove Theorem~\ref{thm:main1}.

\begin{reptheorem}{thm:main1}[restated]
Let $n=|X|$. Algorithm~\ref{alg:main} gives the exact optimal solution to $\MER(X,t)$, running in $O(n^2\log{n}+ktn+kt^3)$ time. As $k\leq n$, it is also $O(n^2\log{n}+n^2t+nt^3)$.
\end{reptheorem}

\begin{proof}
The correctness is clear because we enumerate over all possible rectangles. We analyze the running time below.

First we consider Algorithm~\ref{alg:valid_pairs}, the subroutine to find all valid pairs. For $P_1$ and $P_0$, finding $L$ and $R$ of \emph{ray} $P_1P_0$ takes $O(n)$ time (Line~\ref{alg:valid_pairs:LR}) and sorting them in the clockwise order respectively takes $O(n\log n)$ time (Line~\ref{alg:valid_pairs:LRsort}). In the subsequent rotation of $P_1P$, updating $P_2$, the points on the left side of $P_1P$ and the points on the right side of $P_1P$ runs in $O(1)$ time if $L$ and $R$ are maintained by linked lists. Hence for each fixed $P_1$, Lines~\ref{alg:valid_pairs:loop_content_start}--\ref{alg:valid_pairs:loop_content_end} takes $O(n)$ time. By enumerating $P_1$, Algorithm~\ref{alg:valid_pairs} runs in a total time of $O(n^2\log n)$.

Next consider Algorithm~\ref{alg:para}. Finding the candidates for $P_3$, $P_4$ and $P_5$ by the selection algorithm in \cite{kiwiel2005floyd} takes $O(n)$ time respectively. Further sorting the candidates for $P_3$, $P_4$ and $P_5$ takes $O(t\log{t})$ time respectively. Enumerating over $P_4$ takes $O(t)$ time. From Line~\ref{alg:para:finding_start} to Line~\ref{alg:para:finding_end}, finding $P_3$ and $P_4$ such that the rectangle determined by $P_1$, $P_2$, $P_3$, $P_4$ and $P_5$ encloses $(n-t)$ or $(n-t+1)$ points takes $O(t)$ time. Hence, each call to Algorithm~\ref{alg:para} takes $O(n+t\log{t} + t^2) = O(n + t^2)$ time.
	
Now we return to the main algorithm (Algorithm~\ref{alg:main}). Finding all the valid pairs takes $O(n^2\log{n})$ time, and there are $kt$ of them. For each single valid pair $(P_1,P_2)$, we call the $\Call{enclose}{}$ function, which runs in $O(n + t^2)$ time. The total runtime of Algorithm~\ref{alg:main} is $O(n^2\log{n}+ktn+kt^3)$. As $k\leq n$, it is also $O(n^2\log{n}+n^2t+nt^3)$.
\end{proof}

\section{A Sampling Algorithm}
\label{sec:approximation}
Although the algorithm given in Section~\ref{sec:timealg} is faster than those in earlier studies, it is still not efficient for large-scale datasets. In this section, we shall present a sampling approach to obtain an approximate solution, which can further reduce the running time when $n$ is large. To illustrate the idea, we need to introduce the dual problem to the $\MER(X,t)$, the definition of he VC dimension and a few related results.

\subsection{Dual Problems}
\begin{definition}[Dual of \MER]
	Given a set $X$ of $n$ planar points and a positive value $\alpha$, we intend to find the maximum number of points $\kappa(X,\alpha)$ covered by any rectangle with area at most $\alpha$.
	\label{maximum number formulation}
\end{definition}

Next we explain the relation between the $\MER(X,t)$ problem and its dual problem, assuming the same planar point set. We retain the assumption of no three collinear points. Below are two corollaries of Lemma~\ref{lem:enclosed points number} to be used later.

\begin{lemma}
	\label{lem:dual relation}
	Let $\alpha = \opt(X,t)$. Then $\kappa(X,\alpha)\in\{n-t, n-t+1\}$.
\end{lemma}

\begin{proof}
	Notice that we already have a rectangle $\sol^\ast(X,t)$ of area $\alpha$ that encloses $(n-t)$ or $(n-t+1)$ points by Lemma~\ref{lem:enclosed points number}, therefore it is only possible that $\kappa(X,\alpha)\geq n-t$. On the other hand, if $\kappa(X,\alpha)\geq n-t+2$, we can move inwards a side on which there at most two points. There remain at least $(n-t)$ points while the area is smaller than $\alpha$, contradicting the area optimality of $\sol^\ast(X,t)$. Therefore $n-t\leq\kappa(X,\alpha)\leq n-t+1$.
\end{proof}

In the same spirit, we can prove the following corollary.
\begin{corollary}
	\label{coro:upper bound}
	For any rectangle with area no more than $\alpha=opt(X,t)$, the number of points it encloses can not exceed $n-t+1$.
\end{corollary}

\subsection{VC Dimension of $(X,\mathcal{R})$}
\label{sec:VC}
We take the definitions of the range space and VC dimension from \cite{har2010complexity}. 

\begin{definition}{\cite[Definition 20.1.1]{har2010complexity}}
\label{def:range}
A range space is a pair $(X,\mathcal{R})$, where $X$ is a ground set and $\mathcal{R}$ is collection of subsets of $X$. The elements of $X$ are points and the elements of $\mathcal{R}$ are ranges.
\end{definition}

\begin{definition}{\cite[Definition 20.1.4]{har2010complexity}}
\label{def:shattered}
For a range space $(X,\mathcal{R})$ and a subset $Y\subseteq X$, the projection of the range space on $Y$ is defined to be $\{\tau\cap Y|\tau\in\mathcal{R}\}$. If the projection $\{\tau\cap Y|\tau\in\mathcal{R}\}$ is the power set of $Y$, then we say $Y$ is shattered by $(X,\mathcal{R})$.
\end{definition}

\begin{definition}
\label{def:VC}
The VC dimension of a range space $(X,\mathcal{R})$ is the maximum cardinality of a shattered subset of $X$.
\end{definition}

In our case, $X$ is the given point set. And each element in $\mathcal{R}$ refers to the set of points in a rectangle, including the vertices and those on the four edges. We need the following lemma to bound the VC dimension of $(X,\mathcal{R})$.

\begin{lemma}
	\label{enclosed parts}
	For a convex 10-gon and a rectangle, there are at most 8 intersection points and at most 4 continuous parts of the boundary of the 10-gon lying inside the rectangle.
\end{lemma}

\begin{proof}
	As the 10-gon is convex, there are at most two intersection points on each side of the rectangle. Besides, if $I_{1}$, $I_{2}$ and $I_{3}$ are three continuous intersection points on the sides of the rectangle, then one of polygon boundary parts $I_{1}I_{2}$ and $I_{2}I_{3}$ must be inside the rectangle and the other must be outside. As there are at most 8 intersection points, the number of parts of the polygon inside the rectangle is at most 4.
\end{proof}

\begin{figure}[!t]
	\centering
	\begin{tikzpicture}[scale=0.8]
		\centering
		\node at (-0.7,-0.3) {$A_1$};
		\node at (0.5,-1.1) {$A_2$};
		\node at (2, -1.3) {$A_3$};
		\node at (4, -0.8) {$A_4$};
		\node at (4.5, 0.5) {$A_5$};
		\node at (4, 1.8) {$A_6$};
		\node at (2.8, 2.5) {$A_7$};
		\node at (1.7, 2.6) {$A_8$};
		\node at (0.6, 2.3) {$A_9$};
		\node at (-0.5, 1.5) {$A_{10}$};
		\draw[dashed]  (-0.7,-0.3) -- (0.5,-1.1) -- (2, -1.3) -- (4, -0.8) -- (4.5, 0.5) -- (4, 1.8) -- (2.8, 2.5) -- (1.7, 2.6) -- (0.6, 2.3) -- (-0.5, 1.5) --cycle;
		\draw (0.6, -2) -- (4.6, 0) -- (3, 2.9) -- (-0.9, 1) --cycle;
	\end{tikzpicture}
	\caption{A convex $10$-gon. A rectangle can intersect the $10$-gon at no more than $8$ points and thus only enclose four continuous sections of the boundary.}
	\label{fig:10gon}
\end{figure}
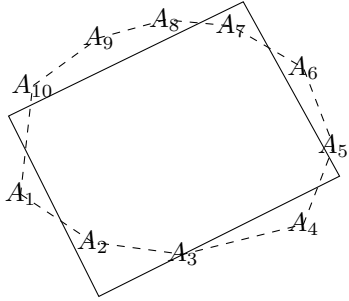

\begin{lemma}
	\label{VC dimension}
	The VC dimension of the range space $(X, \mathcal{R})$ is at most 9.
\end{lemma}

\begin{proof}
	For any given 10 points $A_{1},A_2,\dots,A_{10}$, there are three cases considering their relative positions.
	
	The first case is that $A_{1},A_2,\dots,A_{10}$ are the vertices of a convex 10-gon. Without loss of generality, we suppose that they are in the clockwise order (See Figure~\ref{fig:10gon}). By Lemma~\ref{enclosed parts}, there are at most 4 discontinuous parts of the convex 10-gon boundaary being enclosed by an arbitrary rectangle. Therefore it is impossible for any rectangle to encloses exactly $A_{1}$, $A_{3}$, $A_{5}$, $A_{7}$ and $A_{9}$, otherwise there would be 5 discontinuous parts of the convex 10-gon enclosed by the rectangle.
	
	The second case is that the ten points are not in the position of a convex 10-gon and no three of them are collinear. Then there are at least two convex layers in the onion structure \cite{chazelle1985convex}. It is not possible for any rectangle to contain the outermost layer without containing the inner ones.
	
	The third case is that there are three collinear points $A_{1}$, $A_{2}$ and $A_{3}$. Without loss of generality, we assume $A_{2}$ lie between $A_{1}$ and $A_{3}$. Then it is not possible for a rectangle to enclose $P_{1}$ and $P_{3}$ while not enclosing $P_{2}$.
	
	In any of the three cases, there is at least one subset that can not be the intersection of any rectangle and $\{A_{1},A_2,\dots,A_{10}\}$. It follows from Definition~\ref{def:shattered} and~\ref{def:VC} that the maximum cardinality of any shattered subset can not be 10 or larger in any of the three cases. The VC dimension is therefore at most 9.
\end{proof}

\subsection{Approximation Guarantees}
We next introduce a sampling result for the dual problem of \MER. The following is a discrepancy result for rectangles, slightly modified from an earlier result~\cite[Lemma 9]{de2016covering} for rectangles. We shall discuss its proof at the end of this section.

\begin{lemma}
	\label{sampling}
	Let $s=\min\big(n, \frac{c}{\epsilon^{2}}\frac{n\log{n}}{\kappa(X,\alpha)}\big)$, where $c$ is some absolute constant, and $S$ be a random sample of $X$ with $s$ points. Then with probability at least $(1-1/n)$, it holds for each rectangle $P$ of area at most $\alpha$ that
	$\left| |P\cap X|/n - |P\cap S|/s \right|\leq \epsilon\cdot \kappa(X,\alpha) / n$.
\end{lemma}

\begin{proof}
The proof is almost identical to the case in \cite[Lemma 9]{de2016covering}, except that they concern the axis-aligned rectangle while the rectangle here is of arbitrary orientation. The proof in fact works for a general range space (see \cite[Theorem 2.11]{har2011relative}) with a sample size $|S|$ of the same order, provided that the VC dimension is a constant. Indeed the VC dimension of $(X, \mathcal{R})$ is at most 9 as stated in Lemma~\ref{VC dimension}.
\end{proof}

Now we show our main result of approximation.
\begin{lemma}
	\label{lem:approximation}
	Let $n=|X|$ and assume that $t \leq n/2$. Let $S$ be a random sample of $X$ such that $|S| =\min(n, c\epsilon^{-2}\log n)$, where $c$ is an absolute constant. Moreover, suppose that $t^\prime=(\frac{t}{n} + \epsilon - \frac{\epsilon t}{n})|S|$ is an integer. It holds with probability at least $1-3/n$ that $(n-t)(1-\epsilon)\leq\left|\sol^\ast(S,t^\prime)\cap X\right|\leq n-t+1$.
\end{lemma}

\begin{proof}
	Let $n=|X|$. The rectangle $P^\ast = \sol^\ast(X,t)$ has area $\alpha = \opt(X,t)$ and encloses the most number of points of $X$. By Corollary~\ref{lem:dual relation} it must hold that
	\[n-t\leq|P^\ast\cap X|=\kappa_{\alpha}(X)\leq n-t+1.\]
	Applying Lemma~\ref{sampling} to $P$, we conclude that when
	\[
	|S| = s=\min\left(n,\frac{c}{\epsilon^{2}}\frac{n\log{n}}{\kappa_{\alpha}(X)}\right)\leq\min\left(n, \frac{2c\log{n}}{\epsilon^{2}}\right),
	\]
	with probability at least $1-\frac{1}{n}$,
	\[\frac{|P^\ast \cap S|}{|S|}\geq\frac{|P^\ast\cap X|}{|X|}-\frac{|P^\ast\cap X|}{|X|}\cdot\epsilon=\frac{\kappa_{\alpha}(X)}{n}\cdot(1-\epsilon),\]
	therefore
	\[|P^\ast \cap S|\geq\frac{\kappa_{\alpha}(X)}{n}\cdot(1-\epsilon)s.\]
	Next we let $p^\ast$ denote the rectangle of area at most $\alpha$ that encloses the most number of points of $S$. Then by Corollary~\ref{coro:upper bound},
	\[|p^\ast\cap X|\leq n-t+1.\]
	On the other hand, by the definition of $p^\ast$ we have
	\[|p^\ast\cap S|\geq|P^\ast \cap S|\geq\frac{\kappa_{\alpha}(X)}{n}\cdot (1-\epsilon)s.\]
	Since $\sol^\ast(S,t^\prime)$ is the smallest rectangle enclosing at least $s-t^\prime=\left(1-\frac{t}{n}\right)(1-\epsilon)s$
	points of $S$ and $p^\ast$ encloses at least \[\frac{\kappa_{\alpha}(X)}{n}\cdot(1-\epsilon)s\geq\left(1-\frac{t}{n}\right)(1-\epsilon)s\]
	points of $S$, it must hold that 
	\[\opt(S,t^\prime) \leq \area(p^\ast) \leq \alpha.\]
	Then we can apply Lemma~\ref{sampling} to $\sol^\ast(S,t^\prime)$ as its area is no more than $\alpha$. Overall with probability at least $(1-3/n)$, we have \[\left|\sol^\ast(S,t^\prime)\cap X\right|\geq (n-t)(1-2\epsilon)-\epsilon.\]
	On the other hand, by Corollary~\ref{coro:upper bound}, we conclude that \[\left|\sol^\ast(S,t^\prime)\cap X\right|\leq n-t+1.\]
	Replacing \(\epsilon\) by some \(O(\epsilon)\) with a constant scaling, then \([(n-t)(1-2\epsilon)-\epsilon]\) can be rewritten as \((n-t)(1-\epsilon)\).
\end{proof}

\begin{remark}
	Removing the assumption of no three collinear points, we see that it continues to hold $\kappa_\alpha(X)\geq n-t$ and the proof of Lemma~\ref{lem:approximation} still goes through, that is, it continues to hold that $\left|\sol^\ast(S,t^\prime)\cap X\right|\geq (n-t)(1-\epsilon)$.
\end{remark}

\begin{remark}
	The left inequality (i.e.\@ the lower bound) continues to hold without the assumption of three collinear points. See Subsection~\ref{sec:VC} for discussion.
\end{remark}

In the light of Lemma~\ref{lem:approximation}, we can apply the algorithm in Section~\ref{sec:timealg} to $S$ and obtain an approximate solution.

\begin{reptheorem}{thm:main2}[restated]
	There is a sampling algorithm which, given a set $X$ of $n$ points, with probability at least $(1-3/n)$, finds a rectangle of area at most $\opt(X,t)$ such that the number of enclosed points is between $(n-t)(1-\epsilon)$ and $(n-t+1)$, in time $O\big((\frac{\log{n}}{\epsilon^{2}})^4(\frac{t}{n}+\epsilon)^3+n\big)$.
\end{reptheorem}

\begin{proof}
	The number of enclosed points is guaranteed by Lemma~\ref{lem:approximation} with a rescale of $\epsilon$. Next we analyze the runtime.
	
	First of all, sampling $s$ points takes $O(n)$ time, where $s= O(\frac{\log{n}}{\epsilon^{2}})$. Next, we aim to solve the problem $\MER(S,t^\prime)$ where $t^\prime = (\epsilon+\frac{t}{n}-\frac{\epsilon t}{n})s$, for which we apply Theorem~\ref{thm:main1} and see that the runtime is
	
	\begin{align*}
	&\quad\ O\left(s{t^\prime}^3+s^2t^\prime+s^2\log{s}\right)\\
	&= O\left(\left(\frac{\log{n}}{\epsilon^2}\right)^4 \left(\frac{t}{n}+\epsilon\right)^3
	+\left(\frac{\log{n}}{\epsilon^2}\right)^3\left(\frac{t}{n}+\epsilon\right)
	+\left(\frac{\log{n}}{\epsilon^2}\right)^2\log{\frac{\log{n}}{\epsilon^2}}\right)\\
	&= O\left(\left(\frac{\log{n}}{\epsilon^2}\right)^4 \left(\frac{t}{n}+\epsilon\right)^3\right).\\
	\end{align*}
	The last equation is because
	\[\frac{\left(\frac{\log{n}}{\epsilon^2}\right)^3\left(\frac{t}{n}+\epsilon\right)}{\left(\frac{\log{n}}{\epsilon^2}\right)^2\log{\frac{\log{n}}{\epsilon^2}}}
	=\frac{\frac{\log{n}}{\epsilon^2}\left(\frac{t}{n}+\epsilon\right)}{\log{\frac{\log{n}}{\epsilon^2}}}
	>\frac{\frac{\log{n}}{\epsilon}}{\log{\frac{\log{n}}{\epsilon^2}}}
	=\frac{\frac{\log{n}}{\epsilon}}{O\left(\log{\frac{\log{n}}{\epsilon}}\right)}
	>\Omega(1)
	\]
	and
	\[\frac{\left(\frac{\log{n}}{\epsilon^2}\right)^4 \left(\frac{t}{n}+\epsilon\right)^3}{\left(\frac{\log{n}}{\epsilon^2}\right)^2\log{\frac{\log{n}}{\epsilon^2}}}
	=\frac{\left(\frac{\log{n}}{\epsilon^2}\right)^2 \left(\frac{t}{n}+\epsilon\right)^3}{\log{\frac{\log{n}}{\epsilon^2}}}
	>\frac{\log{n}\cdot\frac{\log{n}}{\epsilon}}{O\left(\log{\frac{\log{n}}{\epsilon}}\right)}>\Omega(1)\]
	Therefore the total runtime is $O\big((\frac{\log{n}}{\epsilon^{2}})^4(\frac{t}{n}+\epsilon)^3+n\big)$.
\end{proof}

\bibliographystyle{elsarticle-num}
\bibliography{reference}

\begin{appendices}
\section{Removing the Assumption of No Three Collinear Points}\label{sec:three_points}

Now we remove the assumption that no three points are collinear. We shall highlight the changes to the algorithm instead of rewriting the pseudocodes. The main change is in the function of $\textsc{Valid-Pairs}$. Inside of having only two points $P_1$ and $P_2$ on the rotating ray, there could be more points on the line $P_1P_2$. Suppose that 
\[L = (A_1,A_2,\dots,A_s,A_{s+1},\dots),\quad R=(B_1,B_2,\dots,B_r,B_{r+1},\dots),\]
where $A_1,\dots,A_s$ and $B_1,\dots,B_r$ are collinear with $P_1$. We shall append the current boundary points (points on the rotating ray) to $L$ and/or $R$ and remove $A_1,\dots,A_s$ from $L$ and $B_1,\dots,B_r$ from $R$, yielding the new left and right sides
\[L' = (A_{s+1},A_{s+2},\dots),\quad R'=(B_{r+1},B_{r+2},\dots).
\]
If one of them contains at most $(t+2)$ points, we shall create a new valid pair $(P_1,A_i)$ for every $i\leq s$ and $(P_1,B_j)$ for every $j\leq r$. Continuing rotating the ray $P_1P$ to the next stopping position, we shall append $(A_1,\dots,A_s)$ to $R'$ and $(B_1,\dots,B_r)$ to $L'$.

Another change concerns finding $H$ consisting of farthest point from $P_1P_2$. Originally we need only to keep $(t+2)$ points, since it will include all points that are at least farther than the $(t+1)$-st farthest point. Without the assumption that no three points are collinear, we need to include all points that are at least farther than the $(t+1)$-st farthest point. Hence we can again maintain a min-heap of size $(t+1)$, so that scanning through all points we obtain a  correct $(t+1)$-st farthest point to $P_1P_2$. Then we need an additional $O(n)$ scan over all points to include all points of the same distance away from $P_1P_2$ as the $(t+1)$-st farthest point. We note that the time complexity of $\textsc{Valid-Pairs}$ remains the same up to a constant.

The last change is inside the function $\textsc{Para}$, instead of checking $J_i$ and $J_{i+1}$ for $Q_3$, we shall need to check all points in $J$ of the same distance as $J_{i+1}$ from the line $P_1P_2$.

We remark that it is no longer true there are at most $O(nt)$ valid pairs for each $P_1$. In the worst case, if $\Theta(n)$ points are collinear, there could be $\Theta(n^2)$ valid pairs. For real-world data, however, this hardly happens. As shown in Remark~\ref{worst v.s. average}, the average number of valid pairs is significantly smaller than the worst case bound.
\section{The Definition of Convex Layers }\label{sec::appendix}
\begin{definition}[\cite{dalal2004counting}]
	Given a finite Euclidean point set $X$, the first convex layer is defined to be the convex hull of $X$. And the $t$-th convex layer is defined to be the convex hull of the rest of $X$, after the points on the first $(t-1)$ convex layers are removed. The collection of these convex layers is called the onion of $X$. And the size of this collection is defined as the convex depth of $X$. When there is no remaining point after those on the first ${(t-1)}$ convex layers are removed, the $t$-th and subsequent convex layers are simply the empty set by definition.
\end{definition}
\end{appendices}
\end{document}